\documentclass[5p]{elsarticle} 
\pdfoutput=1 
\usepackage{graphicx,color,amssymb,amsmath,amsfonts,tikz,psfrag}
\usepackage{amsmath,amssymb,amsfonts,psfrag} 
\usepackage{graphicx}
\usepackage{caption}
\usepackage{subcaption}
\usepackage{subfig}
\usepackage{multi row}
\usepackage{array}
\usepackage{multicol}
\usepackage{sidecap}
\usepackage{tikz}
\usetikzlibrary{shapes,arrows}
\usepackage{enumerate}
\newtheorem{theorem}{Theorem}
\newtheorem{lemma}[theorem]{Lemma}
\newdefinition{remark}{Remark}
\newproof{proof}{Proof}
\newproof{pot}{Proof of Theorem \ref{thm2}}

\def\WVP{\mbox{\it{WVP}}}

\def\SPT{\mbox{\it{SPT}}}
\def\SP{\mbox{\it{SP}}}

\def\P{\mbox{${P}$}}

\def\SeeThrough{\mbox{\textsc{See-Through}}}

\makeatletter
\def\ps@pprintTitle{%
  \let\@oddhead\@empty
  \let\@evenhead\@empty
  \def\@oddfoot{\reset@font\hfil\thepage\hfil}
  \let\@evenfoot\@oddfoot
}
\makeatother

\begin{document}
\title{Weak Visibility Queries of Line Segments \\ in Simple Polygons and Polygonal Domains}

\author[add1]{Mojtaba Nouri Bygi}
\ead{nouribygi@ce.sharif.edu}
\author[add1,add2]{Mohammad Ghodsi}
\ead{ghodsi@sharif.edu}

\address[add1]{Computer Engineering Department, Sharif University of Technology, Iran}
\address[add2]{School of Computer Science, Institute for Research in Fundamental Sciences (IPM), Iran}

  \begin{abstract}
In this paper we consider the problem of computing 
the weak visibility polygon of any query line segment $pq$ (or $\WVP(pq)$) inside a given 
polygon $\P$. Our first non-trivial algorithm runs in simple polygons
and needs $O(n^3 \log n)$ time and $O(n^3)$ space in the preprocessing phase to 
report $\WVP(pq)$ of any query line segment $pq$ in time $O(\log n + |\WVP(pq)|)$.
We also give an algorithm to compute the weak visibility
polygon of a query line segment in a non-simple polygon with $h$ pairwise-disjoint polygonal obstacles 
with a total of $n$ vertices. Our algorithm needs $O(n^2 \log n)$ time and $O(n^2)$ 
space in the preprocessing phase and computes $\WVP(pq)$ in query time of $O(n\hbar \log n + k)$, in which
$\hbar$ is an output sensitive parameter of at most $\min(h,k)$,
and $k = O(n^2h^2)$ is the output size. This is the best query-time result 
on this problem so far.
\end{abstract}

\maketitle

\section{Introduction \label{sec:intro}}

Two points inside a polygon are {\em visible} to each other if their connecting segment remains
completely inside the polygon. {\em Visibility polygon} $VP(q)$ of a point
$q$ in a simple polygon $\P$ is the set of $\P$ points that are visible from $q$. 
The visibility problem has also been considered for line segments.
A point $v$ is said to be {\em weakly visible} to a line segment $pq$ if there exists
a point $w \in pq$ such that $w$ and $v$ are visible to each other.
The problem of computing the {\em weak visibility polygon} (or $\WVP$) of $pq$ inside a polygon $\P$ is 
to compute all points of $\P$ that are weakly visible from $pq$.

If $\P$ is a simple polygon, $\WVP(pq)$ can be computed in linear time \cite{guibas,toussainta}.
For a polygon with holes, the weak visibility polygon has a complicated structure.
Suri and O'Rourke \cite{suri} 
showed that the weak visibility polygon can be computed in $O(n^2)$ time
if output as a union of $O(n^2)$ triangular regions. They also showed that
$\WVP(pq)$ can be output as a polygon in $O(n^2 \log n + k)$, where $k$ is $O(n^4)$.
Their algorithm is worst-case optimal as there are polygons with holes whose weak visibility 
polygon from a given segment can have $\Omega(n^4)$ vertices.

The query version of this problem has been considered by few.
It is shown in \cite{bose} that a simple polygon $\P$ can be preprocessed 
in $O(n^3 \log n)$ time and $O(n^3)$ space
such that given an arbitrary query line segment inside the polygon, 
$O(k \log n)$ time is required to recover $k$ weakly visible vertices.
This result was later improved by \cite{aronov} in which the preprocessing time and space were 
reduced to $O(n^2 \log n)$ and $O(n^2)$ respectively, at the expense of more query 
time of $O(k \log^2 n )$.
In a recent work, we presented an algorithm to report $\WVP(pq)$
of any $pq$ in $O(\log n + |\WVP(pq)|)$ time by spending $O(n^3 \log n)$
time and  $O(n^3)$ space for preprocessing~\cite{nouri}.
Later, Chen and Wang considered the same problem and, by improving the 
preprocessing time of the visibility algorithm of Bose {\em et al.} \cite{bose},
they improved the preprocessing time to $O(n^3)$ \cite{chen2}.
In another work \cite{nouri2}, we showed that the $\WVP(pq)$ 
can be reported in near optimal time of $O(\log^2 n + |\WVP(pq)|)$, 
after preprocessing the input polygon in time and space of $O(n^2 \log n)$
and $O(n^2)$, respectively.

\subsection{Our results}

In the first part of this paper, we present an algorithm for computing the weak visibility 
polygon of any query line segment in a simple polygons $\P$.
We build a data structure in $O(n^3 \log n)$ time and $O(n^3)$ space that can compute 
$\WVP(pq)$  in $O(\log n + |\WVP(pq)|)$ time for any query line segment $pq$.
A preliminary version of this result appeared in~\cite{nouri}.

In the second part of the paper, we consider the problem of computing $\WVP(pq)$ in polygonal domains.
For a polygon with $h$ holes and total vertices of $n$, our algorithm needs the preprocessing
time of $O(n^2 \log n)$ and space of $O(n^2)$. 
We can compute $WVP(pq)$ in time $O(n\hbar \log n + k)$. 
Here $\hbar$ is an output sensitive parameter of at most $\min(h,k)$,
and $k = O(n^2h^2)$ is the size of the output polygon. 
Our algorithm is an improvement over the previous result of Suri and O'Rourke \cite{suri},
considering the extra cost of preprocessing.

\subsection{Terminologies} \label{sec:pre:decompos}
Let $\P$ be a polygon with total vertices of $n$. Also, let $p$ be 
a point inside $\P$.
The {\em visibility sequence} of a point $p$ is 
the sequence of vertices and edges of $\P$ that are visible from $p$. 
A {\em visibility decomposition} of $\P$ is to partition $\P$ into a set of 
{\em visibility regions}, such that any point inside each region has the same visibility sequence. 
This partition is induced by 
the {\em critical constraint edges},
which are the lines in the polygon each induced by two vertices of $\P$,
such that the visibility sequences of the points
on its two sides are different.

The visibility sequences of two  {\em neighboring} visibility regions which are separated by an edge, 
differ only in one vertex.
This fact is used to reduce the space complexity of maintaining the 
visibility sequences of the regions \cite{bose}.
This is done by defining the {\em sink regions}. 
A sink is a region with the smallest visibility sequence compared to all 
of its adjacent regions.
It is therefore sufficient to only maintain the visibility sequences of the sinks,
from which the
visibility sequences of all other regions can be computed.
By constructing a directed dual graph (see Figure~\ref{fig:f2}) over 
the visibility regions, one can maintain the difference between the visibility sequences of 
the neighboring regions \cite{bose}.

\begin{figure}[h]
  \centering
  \includegraphics[width=1\columnwidth]{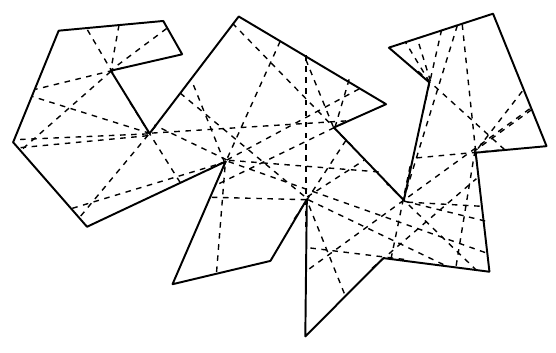} 
  \caption{The visibility decomposition induced by the critical constraints.}
  \label{fig:f2}
\end{figure}

In a simple polygon with $n$ vertices, the number of visibility
and sink regions are respectively $O(n^3)$ and $O(n^2)$~\cite{bose}.
For a non-simple polygon, these numbers are both $O(n^4)$ \cite{zarei}.

\subsection{A linear time algorithm for computing $\WVP$} \label{sec:guibas}
Here, we present the linear algorithm of Guibas {\em et al.}\ for computing
$\WVP(pq)$ of a line segment $pq$ inside $\P$, 
as described in \cite{ghosh}.
This algorithm is used in computing the weak viability polygons in an output sensitive way 
to be explained in Section~\ref{sec:first-alg}.
For simplicity, we assume that $pq$ is a convex edge of $\P$,
but we will show that this can be extended to any line segment in the polygon.

Let $\SPT(p)$ denote the shortest path tree in $\P$ rooted at $p$. 
The algorithm traverses $\SPT(p)$ using a DFS and checks the turn at each vertex
$v_i$ in $\SPT(p)$. If the path makes a right turn at $v_i$, then 
we find the descendant of $v_i$ in the tree with the largest index $j$ (see Figure \ref{fig:guibas}).
As there is no vertex between $v_j$ and $v_{j+1}$,
we can compute the intersection point $z$ of $v_jv_{j+1}$ and $v_kv_i$ 
in $O(1)$ time, where $v_k$ is the
parent of $v_i$ in $\SPT(p)$.
Finally the counter-clockwise boundary 
of $\P$ is removed from $v_i$ to $z$ by inserting the segment $v_iz$.

Let $\P'$ denote the remaining portion of $\P$. We follow the same procedure for
$q$. This time, the algorithm checks every vertex to see whether the path 
makes its first left turn. If so, we will cut the polygon at that vertex in a similar way. 
After finishing the procedure, 
the remaining portion of $\P'$ would be $\WVP(pq)$.

\begin{figure}[h]
  \centering
  \includegraphics[width=1\columnwidth]{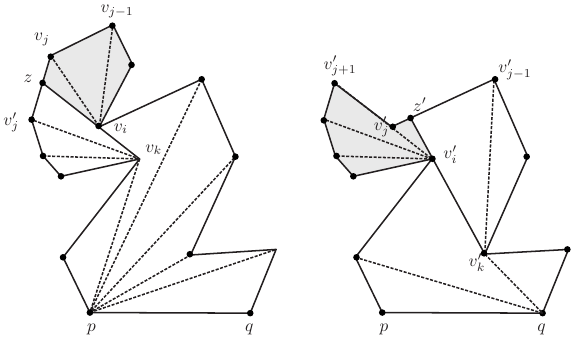} 
  \caption{ The two phases of the algorithm of computing $\WVP(pq)$.
	In the left figure, the shortest path from $p$ to $v_j$ makes a first right turn at $v_i$. In 
	the right figure, the shortest path from $q$ to $v'_j$ makes a first left turn at $v'_i$.}
  \label{fig:guibas}
\end{figure}

\section{Weak visibility queries in simple polygons} \label{sec:first-alg}
In this section, we show how to modify the presented algorithm of Section~\ref{sec:guibas}, 
so that $\WVP$ can be computed efficiently in an output sensitive way.
An important part of this algorithm is computing the shortest path trees. Therefore,
we first show how tom compute these trees in an output sensitive way. Then,
we present a primary version of our algorithm. 
Later, in Section \ref{sec:improve}, we show hot to improve this algorithm.

\subsection{An output sensitive algorithm for computing $\SPT$} \label{sec:spt}
The Euclidean shortest path tree of a point inside a simple polygon of size $n$ 
can be computed in $O(n)$ time \cite{guibas}.
In this section we show how to preprocess $\P$, so that
for any given point $p$ we can report any part of $\SPT(p)$ 
in an output sensitive way.

The shortest path tree $\SPT(p)$ is composed of two
kinds of edges: the {\em primary edges} that connect the root $p$ to its direct visible
vertices, and the {\em secondary edges} that connect two vertices of $\SPT(p)$
(see Figure \ref{fig:pspt}).
We can also recognize two kinds of the secondary edges: a 1st type secondary edge 
(1st type for short) is a secondary edge that is connected to a primary edge, 
and a 2nd type secondary edge (2nd type for short) is a secondary edge that is not connected 
to a primary edge. We show how to store these edges in the preprocessing phase, so that
they can be retrieved efficiently in the query time.

The primary edges of $\SPT(q)$ can be computed 
by using the algorithm of computing the visibility polygons \cite{bose}. 
More precisely, with a preprocessing cost of $O(n^3 \log n)$ time and $O(n^3)$ space, 
a pointer to the sorted list of the visible vertices of a query point $p$ can be 
computed in time $O(\log n)$.

\begin{figure}[h]
  \centering
  \includegraphics[width=1\columnwidth]{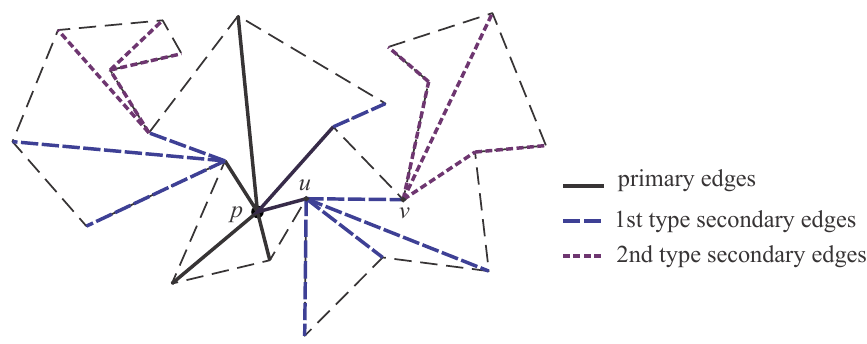}   
  \caption{The shortest path tree of $p$ and its different types of edges:
	the edges that are directly connected to the root $p$ (primary edges),  
	the edges that are connected to the primary edges (1st type secondary edges), 
	and the remaining edges (2nd type secondary edges).}
  \label{fig:pspt}
\end{figure}

For computing the secondary edges of $\SPT$,
in the preprocessing time,
we store all the possible values of the secondary edges of each vertex. 
Having these values, we can detect the appropriate list in the query time
and report the edges without any further cost.

Each vertex $v$ in $\P$ have $O(n)$ potential parents in $\SPT$. For each potential parent
of $v$, it may have $O(n)$ 2nd type edges in $\SPT$. Therefore, for a vertex $v$, 
$O(n^2)$ space is needed to store all the possible combinations of the 2nd type edges. 
Computing and storing these potential edges can be done in $O(n^2 \log n)$ time. 
In the query time, when we arrive at the vertex $v$, we use these 
data to extract the 2nd type edges of $v$ in $\SPT$. Computing these data for all the 
vertices of $\P$ needs $O(n^3 \log n)$ time and $O(n^3)$ space.

The parent of a 1st type edge of $\SPT$ is the root of 
the tree. As the root can be in any of the $O(n^3)$ visibility regions, 
we need to consider all these potential parents to compute
the possible combinations of the 1st type edges of a vertex.
Considering all the regions, the possible first type edges
can be computed in $O(n^4 \log n)$ time and $O(n^4)$ space.

\begin{lemma}
Given a simple polygon $\P$, we can build a data structure of size $O(n^4)$ 
in time $O(n^4 \log n)$, so that for a query point $p$, the 
shortest path tree $\SPT(p)$ can be reported in $O(\log n + k)$ time, where $k$ is 
the size of the tree to be reported.
\end{lemma}

In Section~\ref{sec:improve} we will show how to improve
the processing time and space by a linear factor.

\subsection{Computing the query version of $\WVP$} \label{sec:wvp}
In this section, we use the linear algorithm presented in Section \ref{sec:guibas} for 
computing $\WVP$ of a simple polygon. This algorithm is not output sensitive by itself.
See the example of Figure~\ref{g1}. As stated in Section \ref{sec:guibas}, to compute $\WVP(pq)$,
first we traverse $\SPT(p)$ using DFS and we check the turn at every vertex of $\SPT(p)$.
Consider vertex $v$. As we traverse the shortest path $SP(p, v)$, 
all the children of $v$ must be checked. This can cost $O(n)$ time.
When we traverse $\SPT(q)$, a sub-polygon with $v$ as its vertex will be omitted. 
Therefore, the time spent on processing the children of $v$ in $\SPT(p)$ is redundant.

\begin{figure}[h]
  \centering
  \includegraphics[width=.7\columnwidth]{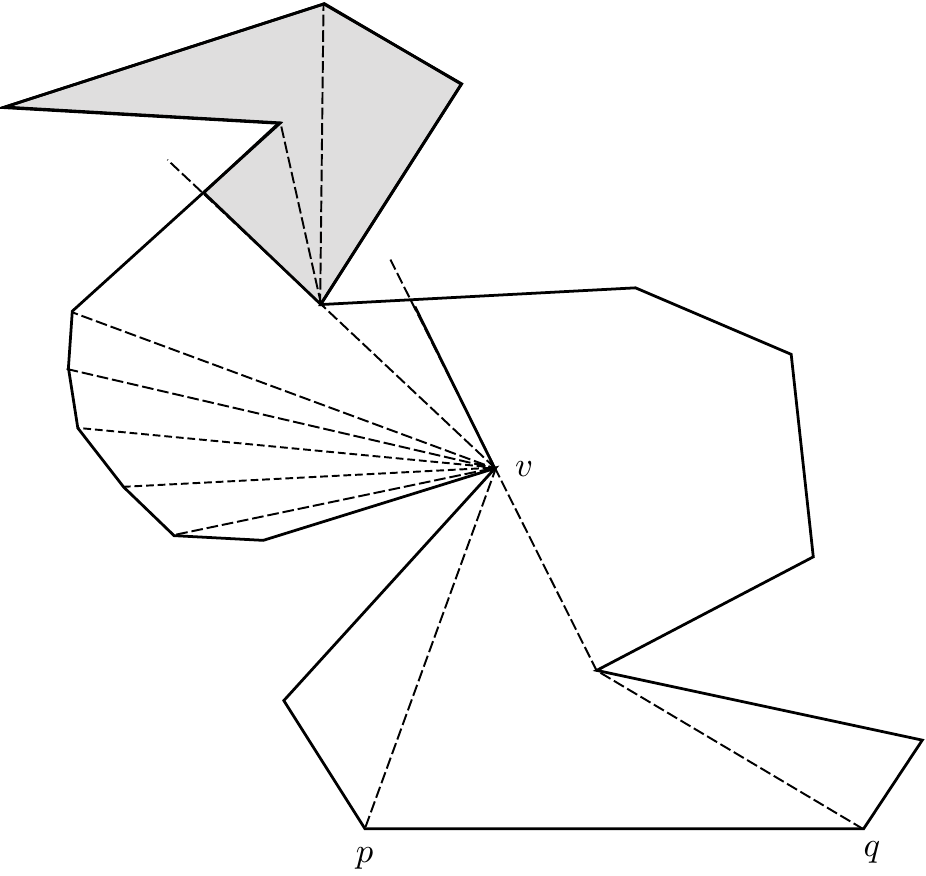} 
  \caption{In the first phase of the algorithm, all the
  children of $v$ in $\SPT(p)$ are processed. These vertices which are not 
  in $\WVP(pq)$ impose redundant $O(n)$ time.} 
  \label{g1}
\end{figure}

To achieve an output sensitive algorithm, 
we build the data structure explained in the previous 
section, so that $\SPT$ of any point inside the polygon can be computed in the query time.
Also, we store some additional information about the
vertices of the polygon in the preprocessing phase.
We say that a vertex $v$ of a simple polygon is {\em left critical} (LC for short) with respect to
a point $p$, if $\SP(p,v)$ makes its first left turn at $v$. In other words, 
each shortest path from $p$ to a non-LC vertex is a convex chain that 
makes only clockwise turns at each node. 
The {\em critical state} of a vertex is whether or not it is LC. 
If we have the critical state of all the vertices of the polygon with respect to a 
point $p$, we say that we have the {\em critical information} of $p$.

The idea is to change the algorithm of Section \ref{sec:guibas} and
make it output sensitive. The outline of the algorithm is
as follows: In the first round, we traverse $\SPT(p)$ using
DFS. At each vertex, we check whether this vertex is left
critical with respect to $q$ or not. If so, we are sure that the
descendants of this vertex are not visible from $pq$, so we
postpone its processing to the time it is reached from $q$,
and check other branches of $\SPT(p)$. Otherwise, we proceed
with the algorithm and check whether $\SPT(p)$ makes a
right turn at this vertex. In the second round,  we traverse $\SPT(q)$ and 
perform the normal procedure of the algorithm.

\begin{lemma} \label{lemma3}
All the traversed vertices in $\SPT(p)$ and $\SPT(q)$ are vertices of $\WVP(pq)$.
\end{lemma}
\begin{proof}
Assume that when we are traversing $\SPT(p)$, we meet $v$ and $v \not\in \WVP(pq)$.
Let $u$ be the parent of $v$ in $SP(pv)$. In this case, $u$ or one of its ancestors 
must be LC with respect to $q$, otherwise the algorithm will detect it as a $\WVP$ vertex. 
Therefore, $v$ cannot be seen while traversing $\SPT(p)$. The same argument can be applied to $\SPT(q)$.
\end{proof}

In the preprocessing phase, we compute the critical
information of a point inside each region, and assign this
information to that region.
In the query time and upon receiving a line segment $pq$, 
we find the regions of $p$ and $q$.
Using the critical information of these two regions, 
we can apply the algorithm and compute $\WVP(pq)$.

As there are $O(n^3)$ regions in the visibility decomposition, 
$O(n^4)$ space is needed to store the critical information of all the vertices. 
For each region, we compute $\SPT$ of a point, 
and by traversing the tree, we update the critical
information of each vertex with respect to that region.  
For each region, we assign an array of size $O(n)$ to store these information.
We also build the structure described 
in Section \ref{sec:spt} for computing $\SPT$ in time $O(n^4 \log n)$ and $O(n^4)$ space.
In the query time, we locate the visibility regions of $p$ and $q$ in $O(\log n)$ time. 
As the processing time spent in each vertex is $O(1)$, by Lemma \ref{lemma3}, 
the query time is $O(\log n + |WVP(pq)|)$.

\begin{lemma} \label{lem:primaryresult}
Using $O(n^4 \log n)$ time to preprocess a simple polygon $\P$ and construct a 
data structure of size $O(n^4)$, it is possible to report $WVP(pq)$ in $O(\log n + |WVP(pq)|)$ time.
\end{lemma}

Until now, we assumed that $pq$ is a polygonal edge.
This can be generalized for any $pq$ in $\P$. 

\begin{lemma} \label{lemma3.1}
Let $pq$ be a line segment inside a simple polygon $\P$. 
We can decompose $\P$ into two sub-polygons $\P_1$ and $\P_2$, such that
each sub-polygon has $pq$ as a polygonal edge. 
Furthermore, the critical information of $\P_1$ and $\P_2$ 
can be computed from the critical information of $\P$.
\end{lemma}
\begin{proof}
We find the intersection points of the supporting line of $pq$ with the
border of $\P$. Then, we split $\P$ into two simple polygons $\P_1$ and $\P_2$, 
both having $pq$ as a polygonal edge.
The visibility regions of $\P_1$ and $\P_2$ are subsets of the visibility regions of 
$\P$. Therefore, we have the critical information and $\SPT$ edges of these regions.
The primary edges of $p$ and $q$ can also be divided to those in $\P_1$ and those
in $\P_2$.
See the example of Figure \ref{fig:split}. 

\begin{figure}[h]
  \centering
  \includegraphics[width=.6\columnwidth]{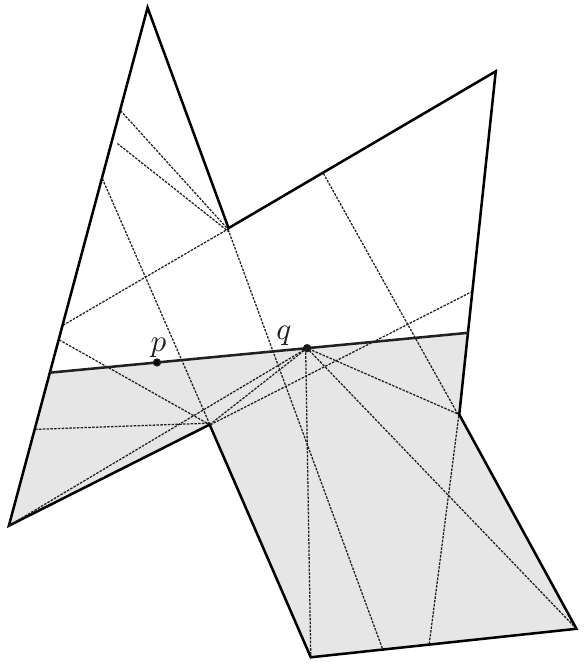}   
  \caption{If the query line segment $pq$ is inside the polygon, we split it along the supporting line of $pq$
  to create two sub-problems. 
  The dotted lines are some of the critical constraints in the polygon. }
  \label{fig:split}
\end{figure}
\end{proof}

\subsection{Improving the algorithm} \label{sec:improve}
In this section we improve the preprocessing cost of Lemma \ref{lem:primaryresult}.
To do this, we improve the parts of the algorithm of Section \ref{sec:wvp} 
that need $O(n^4 \log n)$ preprocessing time and $O(n^4)$ space.
We show that it is sufficient to compute the critical information and the 1st type edges 
of the sink regions (see Section \ref{sec:pre:decompos} for the definition of the 
sink regions). 
For any query point $p$ in a non-sink region, the 1st type edges of $\SPT(p)$ 
can be computed from the 1st type edges 
of the sink regions (Lemma \ref{lem:lemma4}).
Also, the critical information of the other regions
can be deduced from the critical information of the sink regions (Lemma \ref{lem:lemma5}).
As there are $O(n^2)$ sinks in a simple polygon, the processing
time and space of our algorithm will be reduced to $O(n^3 \log n)$ and $O(n^3)$, respectively.

In the query time, if both $p$ and $q$ belong to the sink regions, we have the 
critical information of their regions and we can proceed the algorithm. 
On the other hand, if one of these points is on a non-sink region, 
Lemma~\ref{lem:lemma4} and \ref{lem:lemma5} show that the secondary edges and 
the critical information of that point can be retrieved in $O(\log n + |WVP(pq)|)$ time.

\begin{lemma} \label{lem:lemma4}
Assume that, for a visibility region $V$, the 1st type secondary edges are computed
For a neighboring region that share a common edge with $V$, these edges can be updated
in constant time.
\end{lemma}
\begin{proof}
\begin{figure}[h]
  \centering
  \includegraphics[width=0.65\columnwidth]{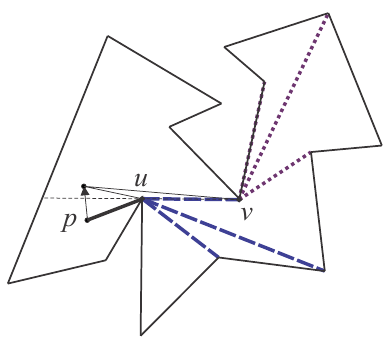}   
  \caption{When $p$ enters a new visibility region, 
  the combinatorial structure of $\SPT(p)$ can be maintained in constant time.}
  \label{fig:h4}
\end{figure}
When a view point $p$ crosses the common border of two neighboring regions, 
a vertex becomes visible or invisible \cite{bose} to $p$. 
In Figure \ref{fig:h4}, for example, when $p$ crosses the
border specified by $u$ and $v$, a 1st type edge of $u$
becomes a primary edge of $p$, and all the edges of $v$ become the 1st type edges.
We can see that no other vertex is affected by this movement.
Processing these changes
can be done in constant time, since it includes the following changes:
removing a secondary edge of $u$ ($uv$), adding a primary edge ($pv$), 
and moving an array pointer (edges of $v$) from the 2nd type edges of $uv$ to the 1st 
type edges of $pv$. Note that we know the exact positions of these elements in their corresponding 
lists. The only edge that involves in these changes (i.e., the edge 
corresponding to the crossed critical constraint), can be identified in the
preprocessing time. Therefore, the time we spent in the query time would be $O(1)$.
\end{proof}

\begin{lemma} \label{lem:lemma5}
The critical information of a point can be maintained
between two neighboring region that share a common edge in constant time.
\end{lemma}
\begin{proof}
\begin{figure}	
	\centering
	\begin{subfigure}[b]{0.23\textwidth}
		\centering
		\includegraphics[width=1\columnwidth]{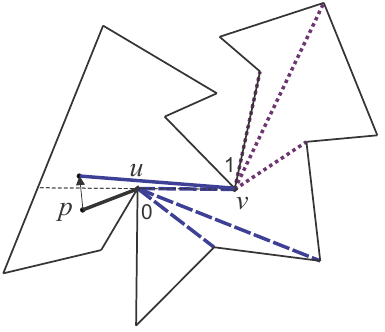}
		\caption{$v$ is LC but not $u$}\label{fig:e2-1a}		
	\end{subfigure}
	\begin{subfigure}[b]{0.23\textwidth}
		\centering
		\includegraphics[width=1\columnwidth]{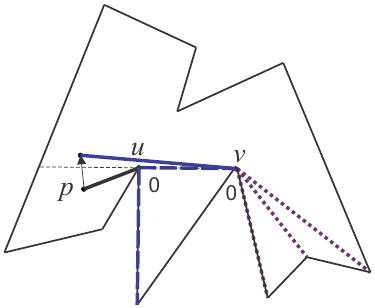}
		\caption{$u$ and $v$ are not LC}\label{fig:e2-1b}		
	\end{subfigure}
	\\
	\begin{subfigure}[b]{0.23\textwidth}
		\centering
		\includegraphics[width=1\columnwidth]{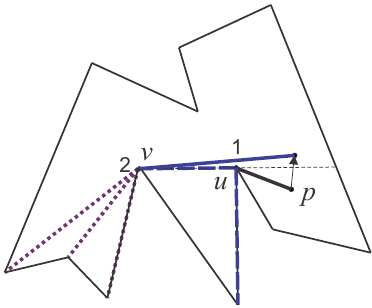}
		\caption{both $u$ and $v$ are LC}\label{fig:e2-1c}		
	\end{subfigure}
	\begin{subfigure}[b]{0.23\textwidth}
		\centering
		\includegraphics[width=1\columnwidth]{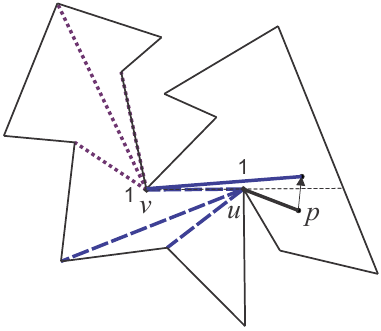}
		\caption{$u$ is LC but not $v$}\label{fig:e2-1d}		
	\end{subfigure}
  \caption{Changes in the critical information of $v$ w.r.t $p$, as $p$ moves between the two regions.}
  \label{fig:e2-1}
\end{figure}

Suppose that we want to maintain the critical
information of $p$, and $p$ is crossing the critical constraint
defined by $uv$. Here, $u$ and $v$ are two 
reflex vertices of $\P$. The only vertices that affect directly by this change
are $u$ and $v$. Depending on the critical states of $u$ and $v$
w.r.t.\ $p$, four situations may occur (see Figure \ref{fig:e2-1}). In the
first three cases, the critical state of $v$ will not change. In
the forth case, however, the critical state of $v$ will change.
Before the cross, the shortest path $\SP(p,v)$ makes a left turn at $u$,
therefore, both $u$ and $v$ are LC w.r.t.\ $p$. However, after the cross, 
$u$ is not on $\SP(p,v)$ and $v$ is no longer LC. 
This means that the critical state of all the children of 
$v$ in $\SPT(p)$ could be changed as well.

To handle these cases, we modify the way the critical information of each vertex 
w.r.t.\ $p$ are stored. At each vertex $v$, we store two additional values: the number 
of LC vertices we met in the path $SP(p,v)$ (including $v$), or its {\em critical number}, 
and {\em debit number}, which is the critical number that is to be propagated in the 
subtree of the vertex. If a vertex is LC, it means that its critical number 
is greater than zero (see Figure \ref{fig:critical-numbers}). Also, if a vertex has 
a non-zero debit number, the critical numbers of all its children must be added by this
number. Computing and storing these additional numbers along the critical information 
will not change the time and space requirements.

Now let us consider the forth case in Figure \ref{fig:e2-1}. When $v$ becomes visible to $p$,
it is no longer LC w.r.t.\ $p$. Therefore, the critical number of $v$ must be changed to $0$,
and the critical number of all the descendants of $v$ in $\SPT(p)$ must be decreased by one.
However, instead of changing the critical numbers of the descendants of $v$, 
we decrease the debit number of $v$ by one, indicating that the critical numbers of 
its descendants in $\SPT(p)$ must be subtracted by one. The actual propagation 
will happen at the query time when we traverse $\SPT(p)$. 
If $p$ moves in the reverse path, i.e., when $v$ becomes 
invisible to $p$, we handle the tree in the same way by adding
$1$ to its debit number, and propagating this addition in
the query time.

\begin{figure}[h]
  \centering
  \includegraphics[width=.9\columnwidth]{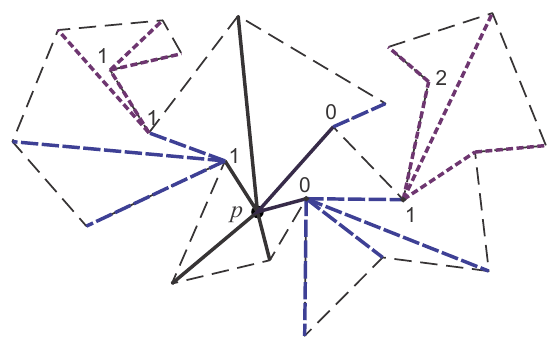}  
  \caption{
The numbers represent the number of left critical vertices met from p in $\SPT(p)$.}
  \label{fig:critical-numbers}
\end{figure}
\end{proof}

In the preprocessing time, we build the dual directed graph of the visibility 
regions. In this graph, every node represents a visibility 
region, and an edge between two nodes corresponds to a gain of one vertex 
in the visibility set in one direction, and a loss in the other direction.
We compute the critical information and 1st type edges of all 
the sink regions. By Lemma \ref{lem:lemma4} and \ref{lem:lemma5}, any two neighboring regions
have the same critical information and secondary edges, except at one vertex. 
We associated this vertex with the edge. 

In the query time, we locate the region containing the point $p$, and follow any path 
from this region to a sink. As each arc represents a vertex that is visible to
$p$, and therefore to $pq$, the number of arcs in the path is 
$O(|WVP(pq)|)$. When traversing the path from the sink back to the region of $p$, 
we update the critical information and the secondary edges of the visible vertices 
in each region. At the original region, we would have the critical information
and the 1st type edges of this region. 
We perform the same procedure for $q$.
Having the critical information and the 1st type edges of $p$ and $q$, we can 
compute $\WVP(pq)$ with the algorithm of Section \ref{sec:wvp}. 
In general, we have the following theorem:

\begin{theorem}
\label{theom:weak_in_simple}
A simple polygon $\P$ can be preprocessed in $O(n^3 \log n)$ time and $O(n^3)$ space such
that given an arbitrary query line segment inside the polygon, 
$\WVP(pq)$ can be computed in $O(\log n + |WVP(pq)|)$ time.
\end{theorem}

\section{Weak Visibility queries in Polygons with Holes}
\label{sec:holes}
In this section, we propose an algorithm for 
computing the weak visibility polygons in polygonal domains. 
Let $\P$ be a polygon with $h$ holes and $n$ total vertices.
Also let $pq$ be a query line segment. 
We use the idea presented \cite{zarei} and convert the non-simple 
polygon $\P$ into a simple polygon $\P_s$. Then, we use the algorithms of computing $\WVP$ in 
simple polygons to compute a preliminary version of $\WVP(pq)$. With some additional work, 
we find the final $\WVP(pq)$.

A hole $H$ can be eliminated by adding two {\em cut-diagonals} connecting a vertex of 
$H$ to the boundary of $\P$. By cutting $\P$ along these diagonals, we will
have another polygon in which $H$ is on its boundary.
We repeat this procedure for all the holes and produce a simple polygon $\P_s$.

\begin{figure}[h]
  \centering
  \includegraphics[width=.6\columnwidth]{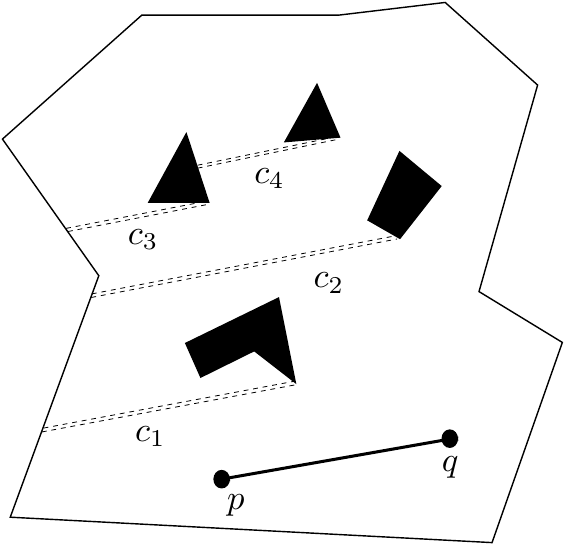}
  \caption{Adding the cut-diagonals to make a simple polygon $\P_s$.}
  \label{fig:wvp3-cuts}
\end{figure}

Let $l$ be the supporting line of $pq$.
For simplicity, we assume that all the holes are on the same side of $l$. Otherwise, we can
split the polygon along $l$ and generate two sub-polygons that satisfy this requirement.
To add the cut-diagonals, we select the nearest point of each hole
to $l$, and perform a ray shooting query from that point in the left direction of $l$,
to find the first intersection with a point of $\P$ (see Figure \ref{fig:wvp3-cuts}). 
This point can be a point on the border of $\P$
or a point on the border of another hole. We select the shooting segment to be the cut-diagonal.
Finding the nearest points of the holes can be done in $O(n\log n)$ time. Also,
performing the ray shooting procedure for each hole can be done in $O(n)$ time.
Therefore, adding the cut-diagonals can be done in total time of $O(n (h + \log n))$.
The resulting simple polygon will have $O(n+2h)$ vertices.
As $h$ is $O(n)$, the number of vertices of $\P_s$ is also $O(n)$.

Having a simple polygon $\P_s$, we compute $\WVP_s(pq)$ in $\P_s$ by using
the algorithm of Section \ref{sec:guibas}.
Next, we add the edges of the polygon that can be seen through the cut-diagonals.
An example of the algorithm can be seen in Figure \ref{fig:w3-steps}.
First, we compute $WVP_s(pq)$ in $\P_s$. Then, for each segment of the cut-diagonals
that can be seen from $pq$, we recursively compute the segments of $\P$ that are visible 
from $pq$ through that diagonal. This leads to the final $\WVP(pq)$.

\begin{figure}[h]
  \centering
  \includegraphics[width=1\columnwidth]{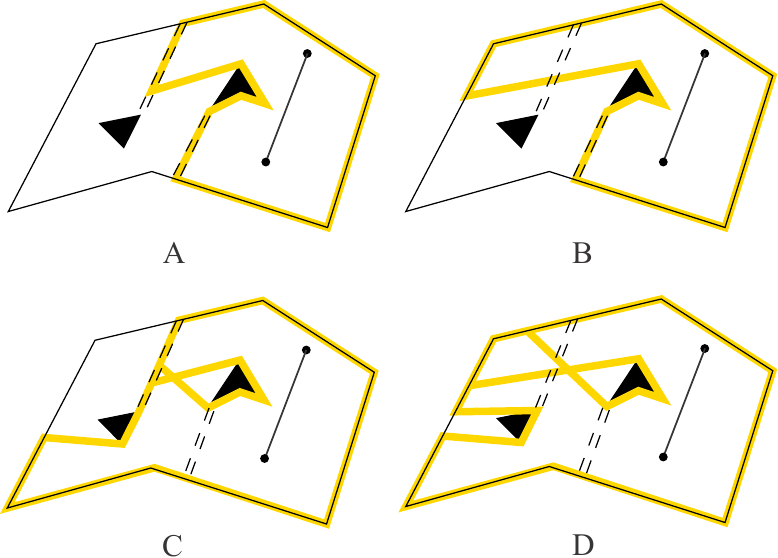} 
  \caption{ Computing $\WVP(pq)$ inside a polygon with holes.}
  \label{fig:w3-steps}
\end{figure}

\subsection{Computing visibility through cut-diagonals}
\label{sec:alg:cuts}

For computing $\WVP(pq)$, we must update $\WVP_s(pq)$ with the edges that 
are visible through the cut diagonals. 
To do this, we define the {\em partial weak visibility polygon}.
Suppose that a simple polygon $\P$ is divided by a diagonal $e$ into two parts, $L$ and $R$.
For a line segment $pq \in R$, we define the partial weak visibility polygon
$\WVP_{L}(pq)$ to be the polygon $\WVP(pq) \cap L$.
In other words, $\WVP_L(pq)$ is the portion of $\P$ that is weakly visible from $pq$ {\em through} $e$.
To compute $\WVP_L(pq)$, one can compute $\WVP(pq)$ by the algorithm of Section \ref{sec:guibas}, 
and then report those vertices in $L$.

\begin{lemma} 
\label{lemma:partial}
Given a polygon $\P$ and a diagonal $e$ which cuts $\P$ into two parts,
$L$ and $R$, for any query line segment $pq \in R$, the partial weak visibility
polygon $\WVP_L(pq)$ can be computed in $O(n)$ time.
\end{lemma}

Lemma \ref{lemma:partial} only holds for simple polygons, but we use its idea 
for our algorithm. Assume that $\P$ has only one hole $H$ and this hole has been eliminated 
by the cut $u_1u_2$. 
Let $v_1v_2$ be another cut which is on the supporting line of $u_1u_2$ and is on the
other side of $H$,  such that $v_1$ is on the border of $H$ and $v_2$ is on the border of $\P$.
We can also eliminate $H$ by $v_1v_2$ and obtain another simple polygon $\P'_s$. 
Now Lemma \ref{lemma:partial} can be applied to the polygon $\P'_s$ and
answer partial weak visibility queries through the cut $u_1u_2$. Following the terminology 
used by \cite{zarei}, we denote this algorithm by $\SeeThrough(H)$. 

By performing the $\SeeThrough(H)$ algorithm once for each hole $H_i$ and assuming that $\P$
has been cut to a simple polygon, we can extend this algorithm to more holes.
This leads to $h$ data structures of size $O(n)$ for storing the simple polygons to
perform Lemma \ref{lemma:partial} for $H_i$.
Using these data structures, we can find the edges of $\P$ that
are visible from $pq$ through the cut-diagonals.

\subsection{The algorithm}
\label{sec:wv3:main-algorithm}
We first add the cut-diagonals to make a simple polygon $\P_s$. 
Then, we compute $\WVP_s(pq)$ and find the set of segments that
are visible from $pq$ in $\P_s$. If a segment $e$ of the cut-diagonal of a hole $H$ is visible 
from $pq$, we use Lemma \ref{lemma:partial} and replace that segment with the partial
weak visibility polygon of $pq$ through that segment. 
We continue this for every cut-diagonal that can be seen from $pq$.
Due to the nature of visibility, this procedure will end. If we have processed $h'$ 
segments of the cut-diagonals, we end up with $h' + 1$ simple polygons of size $O(n)$. It can 
be easily shown that the union of these polygons is $\WVP(pq)$.

Now let us analyze the running time of the algorithm. 
The cut-diagonals can be added in $O(nh+ n\log n)$ time.
Running the algorithm of Theorem \ref{theom:weak_in_simple} in $\P_s$
takes $O(n)$ time. In addition, for each segment of the cut-diagonals that has appeared in $\WVP_s(pq)$, 
we perform the algorithm of Lemma \ref{lemma:partial} in $O(n)$ time.
In general, we have the following lemma:

\begin{lemma}
\label{lemma:wvp3:primary-result}
The time needed to compute $WVP(pq)$ as a set of $h'$ simple polygons
of size $O(n)$ is $O(nh' + n \log n)$, where $h'$ is the
number of cut-diagonals that has been appeared in $\WVP_s(pq)$ 
during the algorithm.
\end{lemma}

\begin{figure}[h]
  \centering
  \includegraphics[width=.6\columnwidth]{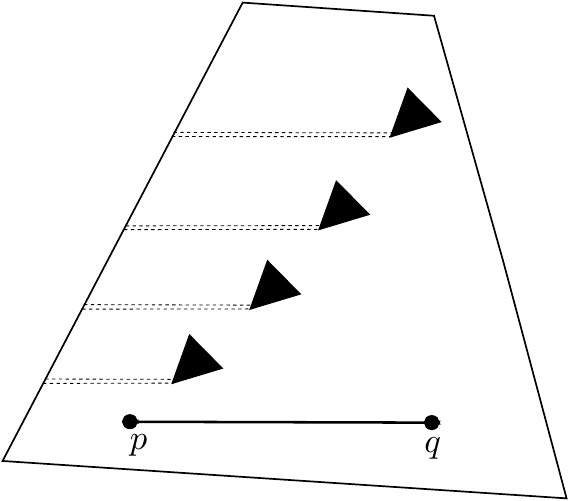} 
  \caption{A polygon with tight bound of $h'$.}
  \label{fig:wvp3-h-bound}
\end{figure}

\begin{lemma}
The upper bound of $h'$ is $O(h^2)$ and this bound is tight.
\end{lemma}
\begin{proof}
We have selected the cut-diagonals in such a way that 
the query line segment $pq$ does not intersect the supporting line of any of 
the cut-diagonals. 
Also, the cut-diagonals do not intersect each other.
Therefore, if $pq$ sees a cut-diagonal $l$ through another cut-diagonal $l'$, then 
$pq$ cannot see $l'$ through $l$. Hence, the upper bound of $h'$ is $O(h^2)$. 
Figure \ref{fig:wvp3-h-bound} shows a sample with tight bound of $h′$.
\end{proof}

\subsection{Improving the algorithm}

In the algorithm of the previous section, we may perform the $\SeeThrough(H)$ algorithm
up to $h$ times for each hole, resulting the high running time of $O(nh^2)$. 
In this section, we show how to change this algorithm and improve the final result. 

A vertex $v$ of the polygon $\P$ can see the line segment $pq$ directly or through the 
cut-diagonals. More precisely, $v$ can see up to $h$ parts of $pq$ through different cut-diagonals. 
These parts can be categorized by the critical constraints that are tangent to the holes
and pass through $v$ and cut $pq$. The next lemma put a limit on the number of these critical 
constraints.

\begin{lemma}
The number of the critical constraints that see $pq$ is $O(n\hbar)$, where
$\hbar = \min(h, |\WVP(pq)|)$ is the number of visible holes from $pq$.
\end{lemma}
\begin{proof}
Let the number of vertices of the hole $H_i$ be $m_i$. 
There are three kinds of constraints: 
\begin{itemize}
\item For each vertex $v$ that is not on the 
border of $H_i$ and is visible to $H_i$, there are at most two critical constraints that 
touch $H_i$ and cut $pq$. Therefore, the total number of these constraints is $O(n\hbar)$.
\item The number of the critical constraints 
induced by two vertices of $H_i$ that cut $pq$ is $O(m)$. We also have $\sum_i m_i = O(n)$.
\item The number of the critical constraints that cut $pq$ and do not touch any hole is 
$O(n)$ \cite{bose}.
\end{itemize}
Putting these together, we can prove the lemma.
\end{proof}

We preprocess the polygon $\P$ so that, in query time, we can efficiently 
find the critical constraints that cut $pq$.
There are $O(n)$ critical constraints passing through each vertex in $\P$.
Therefore, the set of critical constraints can be computed in $O(n^2 \log n)$ time and 
$O(n^2)$ space. 
As the critical constraints passing through a vertex can be treated as a simple polygon 
(see Figure \ref{fig:wvp3-to-simple}),
we build the ray shooting data structure for each vertex in $O(n)$ time and space, so
that the ray shooting queries can be answered in $O(\log n)$ time.
In query time, we find the critical constraints of each vertex that cut $pq$ in $O(c_v \log c_v)$ time,
or in total time of $O(n\hbar \log n)$ for all the vertices. Here $c_v$ is the number of 
constraints that pass through $v$ and cut $pq$.

\begin{figure}[h]
  \centering
  \includegraphics[width=.4\columnwidth]{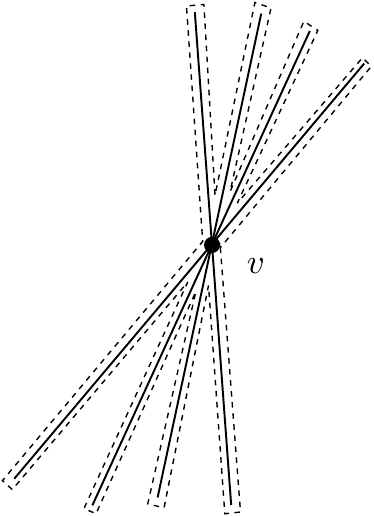} 
  \caption{We can treat the line segments passing through a vertex $v$ as a simple
  polygon (dashed lines) and build the ray shooting data structure  in $O(n)$ time to answer the ray shooting 
  queries.}
  \label{fig:wvp3-to-simple}
\end{figure}

By performing an angular sweep through these lines, we can find the visible 
parts of $pq$ and the visible cut-diagonals from the vertices in $O(n\hbar)$ time.
We store these parts in the vertices, according to the visible cut-diagonal of each part.
Performing this procedure for all the vertices of $\P$, including the vertices of the holes,
and storing the visible parts of $pq$ in each vertex
can be done in $O(n\hbar \log n)$ time and $O(n\hbar)$ space. So, we have the following
lemma:

\begin{figure}[h]
  \centering
  \includegraphics[width=.65\columnwidth]{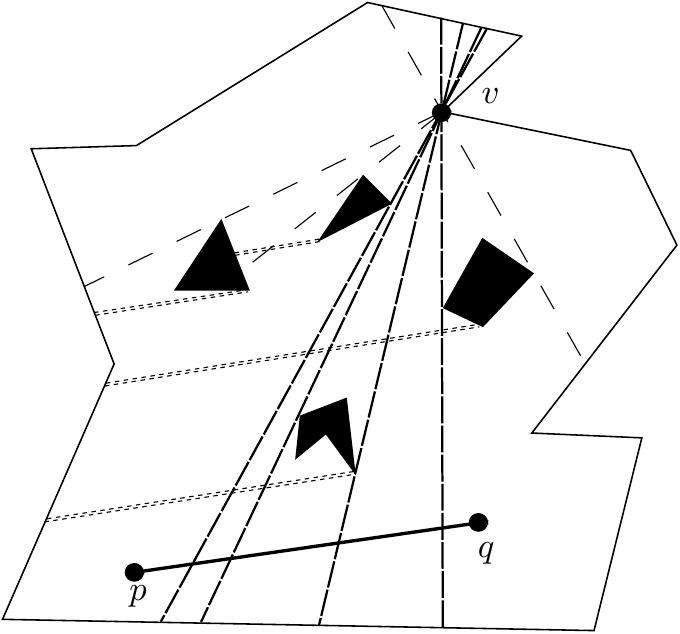} 
  \caption{There are $O(h)$ critical constraints from each vertex 
  of the polygon that hit a cut-diagonal.}
  \label{fig:wvp3-disc-lines}
\end{figure}

\begin{lemma}
Given a polygonal domain $\P$ with $h$ disjoint holes and $n$ total vertices, 
it can be processed into a structure in $O(n^2)$
space and $O(n^2 \log n)$ preprocessing time so that for any query
line segment $pq$, the critical constraints that cut $pq$
can be computed and sorted in $O(n\hbar \log n)$ time, where $\hbar = \min(h, |\WVP(pq)|)$.
\end{lemma}

It the rest of the paper we show that these critical constraints make an arrangement
that can be used to compute $\WVP(pq)$.

We defined $\WVP_s(pq)$ to be the part of $\P$ that can be seen directly from $pq$.
Let $c_i$ be the cut-diagonal of the hole $H_i$.
We define $\WVP_{c_i}$ to be the part of $\P$ that can be seen $pq$ through $c_i$.
It is clear that $\WVP(pq) =  \cup_{i} \WVP_{c_i}(pq) \cup \WVP_s(pq) $. 

Now, we show how to compute $\WVP_{c_i}$. First notice that $\WVP_{c_i}$ is on the
upper half plane of $c_i$.
Let $\P_{c_i}$ be the part of $\P_s$ that is above $c_i$. 
As $pq$ can see $\P_{c_i}$ through different parts of $c_i$, $\WVP_{c_i}$
may not a simple polygon.

Let $D_i$ be the set of critical constraints originating from the vertices of $\P_{c_i}$
that can see $pq$ and directly cut $c_i$, plus the critical constraints
that can see $pq$ and hit the border of $\P_{c_i}$ and cut $c_i$ just before they hit $\P_{c_i}$.
Each critical constraint is distinguished by one or two reflex vertices. 
We call each one of these vertices as the anchor of the critical constraint.
Also, each one of these 
critical constraints may cut the border of $\P_{c_i}$ at most twice. Let $S_i$ be the segments 
on the border of $\P_{c_i}$ resulted from these cuttings. It is clear that 
$|S_i| = O(n + 2h) = O(n)$. 

Let ${\cal L}_i = \cup_{k=1 \ldots i}(S_k \cup D_k)$, and
let ${\cal A}_i$ be the arrangement induced by the segments of ${\cal L}_i$.
We show that ${\cal A}_i$ partitions $\P_{c_i}$ into visible and invisible 
regions.

\begin{lemma}
\label{lemma:e_exists}
For each point $x \in \P_{c_i}$ that is visible from $pq$, there is 
a segment $e$ in ${\cal L}_i$ that can be rotated around its anchor
until it hits $x$, while remaining visible to $pq$.
\end{lemma}

\begin{figure}[h]
  \centering
  \includegraphics[width=.65\columnwidth]{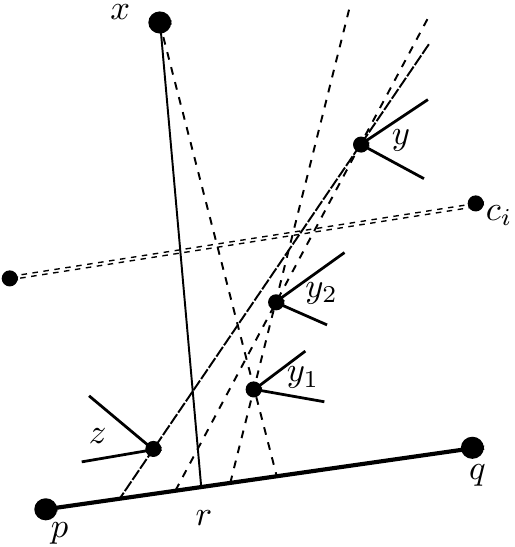} 
  \caption{For each visible point $x \in \P_{c_i}$, there is a  
  critical constraint $yz$ that can be rotated around its anchor $y$ until it hits $x$.}
  \label{fig:wvp3-exist-disc-line}
\end{figure}

\begin{proof}
As $x$ is visible from $pq$, 
it must be visible from some point $r$ of $pq$, such that $xr$ cuts $c_i$
(see Figure \ref{fig:wvp3-exist-disc-line}).
We rotate the segment $xr$ counterclockwise about $x$ until
it hits some vertex $y_1 \in \P$. 
Notice that the case $y_1 = q$ is possible and does not require separate
treatment.
Next, we rotate the segment clockwise about $y_1$ until it hits another vertex $y_2 \in \P$. 
We continue the rotations until the segment reaches one of the endpoints of $c_i$, 
or the lower part of the segment hits a point $z$ of the polygon,
or the segment reaches the end-point $p$. 
Let $y$ be the last point that the segment hits on the upper part of $e$. 
As we only rotate the segment clockwise, this
procedure will end.
It is clear that
$yz$ is a critical constraint in ${\cal L}_i$, and we can reach the point $x$ by rotating
$yz$ counterclockwise about $z$. 
\end{proof}

\begin{lemma}
All the points of a cell $c$ in ${\cal A}_i$ have the same
visibility status w.r.t.\ $pq$.
\end{lemma}

\begin{figure}[h]
  \centering
  \includegraphics[width=.5\columnwidth]{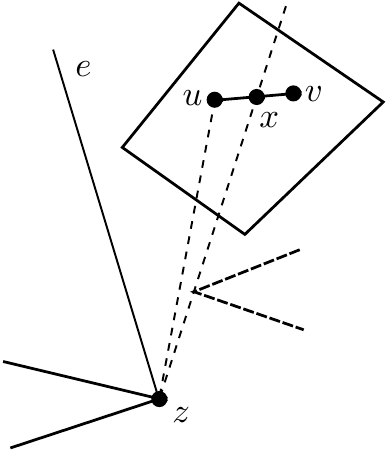} 
  \caption{All the points of a cell have the same visibility status.}
  \label{fig:wvp3-unique-cell}
\end{figure}

\begin{proof}
Suppose that the points $u$ and $v$ are in $c$, and $u$ is visible 
and $v$ is invisible from $pq$. Let $uv$ be the line segment connecting
$u$ and $v$, and $x$ be the nearest point to $u$ on $uv$ that is invisible from $pq$.
According to Lemma \ref{lemma:e_exists}, there is a segment $e \in {\cal L}_i$ with
$z$ as its anchor such that if we rotate $e$ around $z$, it will hit $u$. 
We continue to rotate $e$ until it hits $x$. As $x$ is invisible from $pq$,
$zx$ must be a critical constraint. This means that we have another critical constraint
from a vertex $z \in \P$ that sees $pq$, and it crosses the cell $c$. Thus,
the assumption that $c$ is a cell in ${\cal A}_i$ is contradicted.
\end{proof}

To compute the final $\WVP(pq)$, we have to compute 
$\cup_{i} \WVP_{c_i}(pq) \cup \WVP_s(pq)$. 
$\WVP_s(pq)$ is a simple polygon of size $O(n)$ which can be represented
by $O(n)$ line segments. Also, $\WVP_{c_i}(pq)$ can be represented
by the arrangement of $O(n + 2h + d_i)$ line segments,
where $d_i = |D_i|$. 
It can be easily shown that $\sum_i d_i = O(n\hbar)$. 
Therefore, $\WVP(pq)$ can be represented as the arrangement of 
$O(\sum_{i=1 \ldots \hbar} n +  \sum_{i=1 \ldots h} d_i) = O(n\hbar)$ line segments.

In the next section, we consider the problem of computing the boundary of $\WVP(pq)$.

\subsection{Computing the Boundary of $\WVP(pq)$}
We showed how to output $\WVP(pq)$ as an arrangement of $O(n\hbar)$ line segments.
Here, we show that $\WVP(pq)$ can be output as a polygon in $O(n\hbar \log n + |\WVP(pq)|)$ time.

Balaban \cite{Bal95} showed that by using a data structure of size $O(m)$, one can report
the intersections of $m$ line segments in time $O(m \log m + k)$, where $k$ is the number
of intersections. This algorithm is optimal because at least $\Omega(k)$ time is needed
to report the intersections.
Here, we have $O(n\hbar)$ line segments and reporting all the intersections
needs $O(n\hbar \log n + k)$ time and $O(n\hbar)$ space.
With the same running time, we can classify the edge fragments by using 
the method of Margalit and Knott \cite{MK89}, while reporting the line segment intersections.
We can summarize this in the following theorem:

\begin{theorem}
A polygon domain $\P$ with $h$ disjoint holes and $n$ vertices can be 
preprocessed in time $O(n^2 \log n)$ to build a data structure of size $O(n^2)$, so that the visibility 
polygon of an arbitrary query line segment $pq$ within $\P$ can be computed in 
$O(n \hbar \log n + k)$ time and $O(n \hbar)$ space, where 
$k$ is the size of the output which is $O(n^2 h^2)$ and $\hbar$ is the number of visible holes from $pq$. 
\end{theorem}

\section{Conclusion}
We considered the problem of computing the weak visibility polygon of line segments
in simple polygons and polygonal domains. In the first part of the paper, we presented an algorithm to 
report $\WVP(pq)$ of any line segment $pq$ in a simple polygon of size $n$ in 
$O(\log n + |\WVP(pq)|)$ time, by spending $O(n^3 \log n)$ time preprocessing the polygon 
and maintaining a data structure of size $O(n^3)$. 

In the second part of the paper, we have considered the same problem in polygons with holes. 
We presented an algorithm to compute $WVP(pq)$ of any $pq$ 
in a polygon with $h$ polygonal obstacles 
with a total of $n$ vertices in time $O(n \hbar \log n + k)$
by spending $O(n^2 \log n)$ time preprocessing the polygon and maintaining a data structure of
size $O(n^2)$. The factor $\hbar$ is an output sensitive parameter of size at most
$\min(h, k)$, and $k = O(n^2 h^2)$ is the size of the output.


\begin{thebibliography}{99}





\bibitem{aronov}  B. Aronov, L. J. Guibas, M. Teichmann and L. Zhang.
Visibility queries and maintenance in simple polygons.
{\em Discrete and Computational Geometry}, 27(4):461-483,
2002.

\bibitem{Bal95} 
I.J. Balaban. An optimal algorithm for finding segment intersections. In {\em Proc. 11th
Annu. ACM Sympos. Comput. Geom.}, pages 211-219, 1995.



\bibitem{bose} P. Bose, A. Lubiw, J. I. Munro. Efficient visibility
queries in simple polygons. {\em Computational Geometry:
Theory and Applications}, 23(3):313-335, 2002.


\bibitem{ghosh}
S. K. Ghosh. 
Visibility Algorithms in the Plane.
{\em Cambridge University Press}, New York, NY, USA, 2007.

\bibitem{chen2} D. Z. Chen and H. Wang. Weak visibility queries of line segments in simple polygons.
In {\em 23rd International Symposium, ISAAC}, pages 609-618, 2012.


\bibitem{guibas}
L. J. Guibas, J. Hershberger, D. Leven, M. Sharir, and R. E. Tarjan.
\newblock  Linear time algorithms for visibility and shortest path problems inside triangulated simple polygons.
\newblock {\em Algorithmica}, 2:209-233, 1987.

\bibitem{MK89}
A. Margalit and G.D. Knott. An algorithm for computing the union, intersection or
difference of two polygons. {\em Comput. \& Graph.}, 13:167-183, 1989.


\bibitem{nouri}
M. Nouri Bygi and M. Ghodsi.
Weak visibility queries in simple polygons.
In {\em Proc. 23rd Canad. Conf. Comput. Geom.}, 2011.

\bibitem{nouri2}
M. Nouri Bygi and M. Ghodsi.
Near optimal line segment weak visibility queries in simple polygons.
{\em CoRR}, abs/1309.7803, 2013.


\bibitem{suri}
S. Suri and J. O'Rourke.
Worst-case optimal algorithms for constructing visibility polygons with holes.
In {\em Proc. of the second annual symposium on Computational geometry}, pages 14-23, 1986.

\bibitem{toussainta}
G. T. Toussainta
A linear-time algorithm for solving the strong hidden-line problem in a simple polygon.
{\em Pattern Recognition Letters}, 4:449-451, 1986.

\bibitem{zarei}
A. Zarei and M. Ghodsi.
Efficient computation of query point visibility in polygons with holes.
In {\em Proc. of the 21st Symp. on Comp. Geom.}, pages 314-320, 2005.


\end{thebibliography}
\end{document}